\documentclass{article}
\usepackage{hyperref}
\usepackage{amsmath,mathrsfs}
\usepackage{graphicx}%
\usepackage{amsfonts}%
\usepackage{amssymb}
\usepackage{tikz}
\usepackage{color}

\newtheorem{theorem}{Theorem}
\newtheorem{corollary}{Corollary}
\newtheorem{definition}{Definition}
\newtheorem{example}{Example}
\newtheorem{lemma}{Lemma}
\newtheorem{remark}{Remark}
\newenvironment{proof}[1][Proof]{\textbf{#1.} }{\ \rule{0.5em}{0.5em}}

\begin{document}

\title{Saturability of the Quantum Cram\'{e}r-Rao Bound in Multiparameter Quantum Estimation at the Single-Copy Level}
\author{Hendra I. Nurdin\thanks{Email: h.nurdin@unsw.edu.au} \\ School of Electrical Engineering and Telecommunications \\ University of New South Wales }
\date{
}
\maketitle 

\begin{abstract}
The quantum Cram\'{e}r-Rao bound (QCRB) as the ultimate lower bound  for precision in quantum parameter estimation is only known to be saturable in the multiparameter setting in special cases and under conditions such as full or average commutavity of the symmetric logarithmic derivatives (SLDs) associated with the parameters. Moreover, for general mixed states, collective measurements over infinitely many identical copies of the quantum state are generally required to attain the QCRB. In the important and experimentally relevant single-copy scenario, a necessary condition for saturating the QCRB in the multiparameter setting for general mixed states is the so-called partial commutativity condition on the SLDs. However, it is not known if this condition is also  sufficient.  This paper establishes necessary and sufficient conditions for saturability of the multiparameter QCRB in the single-copy setting in terms of the commutativity of a set of projected SLDs and the existence of a unitary solution to a system of nonlinear partial differential equations. New necessary conditions that imply partial commutativity  are also obtained, which together with another condition become sufficient. Moreover, when the sufficient conditions are satisfied an optimal measurement saturating the QCRB can be chosen to be projective and explicitly characterized. An example is developed to illustrate the case of a multiparameter quantum state where the conditions derived herein are satisfied and can be explicitly verified.  \textcolor{blue}{Published as [H. I. Nurdin, IEEE Control Systems Lett., vol.  8, pp. 376 - 381,  2024 (DOI: 10.1109/LCSYS.2024.3382330)]. Version 5 appends a corrigendum that fixes a gap in Condition 2 of Theorem 2 of this paper.}
\end{abstract}

\section{Introduction}
Estimation of unknown parameters of interest from noisy observations that contain information about the parameters is an important problem originating in statistics that is of fundamental importance and have  wide utility in various areas of science and engineering. In systems and control, parameter estimation is central to important topics in the field such as stochastic modeling and system identification \cite{fan2008nonlinear,Ljung99}.

In physical systems, information about the parameters are typically obtaining by performing some measurements on a system and constructing an estimator for the parameters based on the measurement results. In quantum systems, there is an inherent fundamental noise always present, quantum noise, that persists even if all classical noise sources can be completely eliminated. Thus there has been much interest in parameter estimation when the limiting factor is quantum noise and to achieve the ultimate estimation precision physically possible, typically in the mean square sense. Quantum parameter estimation theory originated
in the pioneering works of Helstrom, Holevo and Belavkin in the '60s and '70s and in recent years has attracted more attention as one of theoretical underpinnings for the field of quantum metrology; see \cite{Liu20} for a recent survey. This field aims to exploit quantum effects and quantum devices to perform measurements more accurately for emerging applications such as quantum sensing and imaging, in various physical platforms such as quantum optics, photonics and cold atoms \cite{Liu20,LZCWH22,Barbieri22}. In quantum systems and control, quantum parameter estimation is foundational for quantum system identification; see, e.g., \cite{NG22} and the references therein.

The quantum Cram\'{e}r-Rao bound (QCRB) in the single and multiparameter setting sets the ultimate precision in the mean square sense with which parameters encoded in the quantum state of a quantum system can be estimated using quantum measurements. When there is only a single parameter, there always exists a quantum measurement that saturates this bound. However, in the multiparameter setting, with two or more parameters to be estimated, this is no longer the case. The well-known reason is that the measurements that achieve the ultimate precision for the different parameters may not in general be compatible with one another, requiring the measurement of non-commuting observables. In recent years, multiparameter quantum estimation theory has been gaining increased attention and there have appeared several survey papers in the literature that give an overview of the key results and state-of-the-art in this research area, see, e.g., \cite{SBD16,Liu20,ABGG20,DGG20,HRZ22}. 

From an experimental point-of-view, {\em single-copy} scenarios of multiparameter quantum estimation are of particular importance because they are easier to implement in a laboratory. Single-copy here means that measurement is only performed on a single copy of the quantum state of interest and a parameter estimate is furnished based on this single-copy measurement. In general, saturating the QCRB in the multiparameter scenario for arbitrary parameterized quantum states requires performing collective measurements on infinitely many identical copies of the quantum state \cite{RJD16}. 

In the single-copy case that is of interest in the present work, a well-known result is that the QCRB can be saturated in multiparameter estimation on pure quantum states,  provided that an average commutativity condition on the symmetric logarithmic derivatives (SLDs) associated with the parameters is satisfied; see Section \ref{sec:prelim} for details. For mixed states that are full rank the QCRB can be saturated if and only if the SLDs for the different parameters are mutually commuting \cite[\hbox{ref. [5]  and Appendix B.1}]{SYH20}. For general mixed states that are between these two extremes, the work \cite{Yang19} derived another type of commutativity condition on the SLDs, called the {\em partial commutativity} condition, and show that this condition is {\em necessary} for saturation of the QRCB in the multiparameter and single-copy case. The paper also derives necessary and sufficient conditions for general quantum measurements described by positive operator-valued measures (POVMs) to saturate the QCRB, however it does not establish the existence of such POVMs for a given quantum state. This single-copy result generalizes an analogous result in \cite{Pezze17} for the special case of pure states and projective measurements. An alternative set of necessary and sufficient conditions for a POVM to saturate the QRRB, but expressed in terms of both the POVM and the associated estimator $\hat{\theta}$, is given in \cite[Appendix B.1]{SYH20}. Whether the partial  commutativity condition is sufficient has up to now been unknown. Obtaining necessary and sufficient conditions for saturating the QRCB in terms of the SLDs, if they exist,  is of significance interest in practice. For instance, they can be used to determine if saturation can be achieved with fewer experimental resources.  

This paper establishes necessary and sufficient conditions for saturability of the  multiparameter QCRB in the single-copy setting. Then new necessary conditions that imply partial commutativity are obtained, building on the approach of \cite{Yang19}. They become sufficient with the addition of another condition (Theorem 6). When sufficient conditions for saturability are met, the proof gives an explicit construction of a measurement that saturates the QCRB, which turns out can be chosen to be projective.

The paper is structured as follows. Section \ref{sec:prelim} gives a brief overview of quantum parameter estimation theory, including a statement of the QCRB and definitions of POVMs and SLDs. Section \ref{sec:existing} reviews existing results on saturability of the QCRB in the single-copy case as well as in the multi-copy setting. This section also recalls some of the key results from \cite{Yang19} that are relevant for deriving Theorem \ref{thm:main}.  Section \ref{sec:main} then states and derives the main results of the paper followed by a discussion and an example. Finally, Section \ref{sec:conclu} gives a conclusion for the paper and directions for future work. 

\vspace{0.5cm}
\noindent \textbf{Notation.} $\mathbb{R}^n$ and $\mathbb{C}^n$ denote the set of real and complex vectors, respectively, whose elements are represented as a column vector unless stated otherwise, and $i=\sqrt{-1}$. The conjugate of a complex number $c$, its real part and imaginary part will be denoted by $\overline{c}$, $\Re\{c\}$ and $\Im\{c\}$, respectively. The transpose of a matrix $X$ is denoted by $X^{\top}$ and the adjoint of an operator $X$ on a Hilbert space $\mathcal{H}$ or the conjugate transpose of a complex matrix $X$ is denoted by $X^{\dag}$. A vector in a complex Hilbert space will be denoted by the ket $|x\rangle$ and its conjugate transpose by the bra $\langle x|$. The trace of a square matrix $X$ is denoted  by $\mathrm{tr}(X)$. The direct sum of two vector spaces $V_1$ and $V_2$ is denoted by $V_1 \oplus V_2$. For a vector space $V$ and an operator $X$ on a vector space containing $V$, $\left. X \right|_V$ denotes the restriction of $X$ to $V$. For any Hermitian matrix $X$, $X \geq  0\, (>0)$ denotes that $X$ is positive semidefinite (positive definite) while $A \geq B\, (A > B)$ for any two Hermitian matrices of the same dimension denotes that $A-B \geq 0$ ($A-B>0)$. For two square matrices $X$ and $Y$, $[X,Y]=XY-YX$ and $\{X,Y\}=XY+YX$ are their commutator and anti-commutator, respectively. An $n \times n$ identity matrix will be denoted by $I_n$ or simply by $I$ if its dimension can be inferred from the context. Similarly, $0_{m \times n}$ will denote a zero matrix of dimension $m \times n$ with the subscript dropped if the dimension can be inferred from context. The expectation operator will be denoted by $\mathbb{E}\left[\cdot\right]$ and the expectation of a random variable $X$ by $\mathbb{E}[X]$.
 
\section{Preliminaries}
\label{sec:prelim}
Consider a finite dimensional quantum system with Hilbert space $\mathcal{H}$ that is of a finite dimension $n_s$. Let $\rho_{\theta}$ be a density operator on $\mathcal{H}$ that is parameterized by an unknown parameter vector $\theta =(\theta_1,\ldots,\theta_p)^{\top} \in \Theta$ with $p$ elements, where $\Theta \subseteq \mathbb{R}^p$ is the parameter space as an open set in $\mathbb{R}^p$. It is assumed throughout that $\rho_{\theta}$ depends smoothly on $\theta$. Let the null space of $\rho_{\theta}$ be denoted by $\mathcal{H}_{0,\theta}=\{|\psi \rangle \in \mathcal{H} \mid \rho_{\theta} |\psi\rangle = 0\}$. It is assumed throughout the paper that $\mathcal{H}_{0,\theta}$ has a fixed dimension $r_0$ for all $\theta \in \Theta$ (i.e., $r_0$ is independent of $\theta$) and it has a set of orthonormal basis vectors $\mathcal{B}_{0,\theta}=\{ \phi_{1,\theta},\ldots,\phi_{r_0,\theta}\}$. Let $\mathcal{H}_{+,\theta}$ be the support of $\rho_{\theta}$ defined by $\mathcal{H}_{+,\theta} =\{\rho_{\theta} |\psi \rangle \mid | \psi \rangle \in \mathcal{H}\}$. That is, the support of $\rho_{\theta}$ coincides with its range.   Since $\rho_{\theta}$ is self-adjoint, the range and support of $\rho_{\theta}$ are orthogonal by the null space-range decomposition of linear algebra, therefore we have the direct sum decomposition $\mathcal{H}=\mathcal{H}_{+,\theta} \oplus \mathcal{H}_{0,\theta}$. Since the dimension of $\mathcal{H}_{0,\theta}$ is fixed, so is the dimension of $\mathcal{H}_{+,\theta}$. This dimension is denoted by $r_{+}= n_s - r_0$.  Moreover, we shall take  $\mathcal{B}_{+,\theta}=\{ \psi_{1,\theta},\ldots,\psi_{r_+,\theta}\}$ as an orthonormal basis for $\mathcal{H}_{+,\theta}$. Based on the stated assumptions, we have:
\begin{equation}
\rho_{\theta} = \sum_{k=1}^{r_{+}} q_{k,\theta} |\psi_{k,\theta}\rangle \langle \psi_{k,\theta}|, \label{eq:rho-decomposition}
\end{equation}
for real numbers $q_{k,\theta}>0$ satisfying $\sum_{k=1}^{r_+} q_{k,\theta}=1$. 

Let $P_{+,\theta}$ denote the projection operator onto $\mathcal{H}_{+,\theta}$ and $P_{0,\theta}$ be the projection onto $\mathcal{H}_{0,\theta}$.  Throughout the paper, all operators on $\mathcal{H}$, such as $\rho_{\theta}$, $L_{\theta_j}$, etc, will often be implicitly represented as complex matrices with respect to the full basis $\mathcal{B} = \mathcal{B}_{0,
\theta} \cup \mathcal{B}_{+,
\theta}$, without further comment. Also, it will be useful to express operators $O$ on $\mathcal{H}$ in the block form:
\begin{equation}
O=\left[\begin{array}{cc} O_{++} & O_{+0} \\ O_{0+} & O_{00} \end{array} \right], \label{eq:decom}
\end{equation}
where $O_{jk}=P_{j,\theta} O P_{k,\theta}$ for $j,k \in \{+,0\}$. If $O$ is an observable $O=O^{\dag}$ then $O_{++}^{\dag}=O_{++}$, $O_{0+}=O_{+0}^{\dag}$ and  $O_{00}^{\dag}=O_{00}$. In this representation and block form, note that
\begin{align}
\rho_{\theta} =\left[ \begin{array}{cc} \rho_{\theta,++} & 0 \\ 0 & 0\end{array} \right], \label{eq:rho-rep}
\end{align}
where $\rho_{\theta,++}>0$, and diagonal in the basis $\mathcal{B}_{+,\theta}$. 

The quantum system is prepared in the state $\rho_{\theta}$ and a general POVM measurement with a discrete and finite number of real outcomes is performed on it. The POVM  will be described by the set of operators $\{E_k;k=1,2,\ldots,M\}$ for some integer $M \geq 2$, where the $E_k$'s are non-zero positive semidefinite operators on $\mathcal{H}$ that satisfy $\sum_{k=1}^M E_k = I_{n_s}$. Each element $E_k$ of a POVM will be referred to as a {\em POVM operator} (also referred to as an effect operator in the literature). Each $E_k$ corresponds to a distinct measurement outcome that is indexed by $k$ and takes on a real value $\mu_k \in \mathbb{R}$. The probability of obtaining a measurement result $\mu_k $ is given by $p_{k,\theta} = \mathrm{tr}(\rho_{\theta} E_k)$.  A POVM is said to be {\em projective} if it corresponds to a projective measurement. In this case   all the POVM operators are mutually commuting projection operators, $E_k^2=E_k$ and $[E_k,E_l]=0$ for all $k,l=1,\ldots,M$.

Given a random measurement outcome $\mu_k$ and unknown parameter value $\theta$, an estimator $\hat{\theta}_k$ of $\theta$ as a random variable that is a function of $\mu_k$ can be constructed. The estimator is given by $\hat{\theta}_k=f(\mu_k)$ for some (Borel measurable) function $f: \mathbb{R} \rightarrow \mathbb{R}^p$, and  it is unbiased, meaning that $\mathbb{E}[\hat{\theta}] =\theta$. 
The covariance matrix of the estimator, denoted by $\Sigma$, is a real symmetric $p \times p$ matrix given  by $\Sigma=\mathbb{E}[(\hat{\theta}-\theta)(\hat{\theta}-\theta)^{\top}]$. A central result in quantum estimation theory is the quantum Cr\'{a}mer-Rao bound (QCRB), which states that the covariance matrix of any unbiased estimator of $\theta$ satisfies the matrix inequality:$\Sigma \geq F_{\theta}^{-1}$, where $F_{\theta}$ is a $p \times p$ real symmetric matrix known as the quantum Fisher information matrix with $F_{\theta}=[F_{\theta,jk}]_{j,k=1,\ldots,p}$ and matrix elements given by:
$$
F_{\theta,jk} =  \mathrm{tr}(\rho_{\theta} \{L_{\theta_j},L_{\theta_k}\})/2,
$$
where $L_{\theta_j}$ is an observable on the system Hilbert space (represented by an $n_s\times n_s$ Hermitian matrix) called the symmetric logarithmic derivative (SLD) with respect to the component $\theta_j$ that is defined via the relationship:
$$
\frac{\partial \rho_{\theta}}{\partial \theta_j} = \frac{1}{2}(L_{\theta_j}\rho_{\theta} + \rho_{\theta} L _{\theta_j}).  
$$
Based on the decomposition \eqref{eq:rho-decomposition} for $\rho_{\theta}$ and \eqref{eq:decom} for $L _{\theta_j}$, the equation for the SLD reduces to
\begin{align}
\begin{split}
&\frac{1}{2}(L_{\theta_j,++} \rho_{\theta,++} + \rho_{\theta,++}L_{\theta_j,++})=P_{+,\theta} \frac{\partial \rho_{\theta}}{\partial \theta_j} P_{+,\theta}\\
&\frac{1}{2}\rho_{\theta,++} L_{\theta_j,+0} = P_{+,\theta} \frac{\partial \rho_{\theta}}{\partial \theta_j} P_{0,\theta}\\
&P_{0,\theta} \frac{\partial \rho_{\theta}}{\partial \theta_j} P_{0,\theta} =0.
\end{split} \label{eq:SLD-cond}
\end{align}
Note that the last equation above is implied by \eqref{eq:rho-decomposition}. Also, $L_{\theta_j,00}$ for $j=1,\ldots,p$ are not determined by the above equations and can be specified arbitrarily as long as they are self-adjoint. 

The QCRB is said to be saturated at a parameter value $\theta$ when equality holds, $\Sigma = F_{\theta}^{-1}$.  When the QCRB is saturated then there exists a POVM $\{E_k\}_{k=1,\ldots,M}$ such that the classical Fisher information matrix of the discrete  probability distribution  $\{p_{k,\theta}=\mathrm{tr}(\rho_{\theta} E_k)\}_{k=1,\ldots,M}$, given by
$F_{\theta,\mathrm{c}} = [F_{\theta,\mathrm{c},lm}]_{l,m=1,\ldots,p}$ with
$$
F_{\theta,\mathrm{c},lm} = \mathbb{E}\left[\frac{\partial \ln p_{k,\theta}}{\partial \theta_l} \frac{\partial \ln p_{k,\theta}}{\partial \theta_m} \right]=\sum_{k=1}^M p_{k,\theta} \frac{\partial \ln p_{k,\theta}}{\partial \theta_l} \frac{\partial \ln p_{k,\theta}}{\partial \theta_m} 
$$
equals the quantum Fisher information matrix, $F_{\theta,\mathrm{c}}=F_{\theta}$. 

For any real positive definite matrix $G$, called a cost matrix, one can associate the scalar bound $\mathrm{tr}(G \Sigma) \geq \mathrm{tr}(G F_{\theta}^{-1})$. The left hand side of this scalar bound gives the variance of some linear combination of elements of the estimator $\hat{\theta}$, the linear combination being determined by $G$. When the QCRB is saturated then the scalar cost is saturated for any choice of the cost matrix $G$, $\mathrm{tr}(G \Sigma_{\theta}) = \mathrm{tr}(G F_{\theta}^{-1})$ \cite{Vidrighin14,DGG20,RJD16,Yang19}.  

In the multiparameter setting with $p>1$, the QCRB will not be saturable in general. At the single-copy level, as introduced earlier, only a single copy of a quantum system that is prepared in the state $\rho_{\theta}$ is available. Measurement is performed on this single copy and the measurement outcome is used to compute an estimator for $\theta$. In the multi-copy setting, $K$ identical copies of the system can be used, each copy prepared in the state $\rho_{\theta}$, and two types of measurements can be performed, {\em separable} and {\em collective} measurements. A separable measurement involves only performing measurements on each copy independently (the $K$ copies are not made to interact) and using the independent measurement results to construct an estimator. In  a collective measurement, the $K$ copies are initially coupled through some quantum operation and this is then followed by collective measurements on the $K$ copies, possibly involving the measurement of joint observables on the $K$-copies, which is experimentally challenging; see \cite{RJD16,DGG20,HRZ22}. In general, collective measurements may be required to asymptotically saturate the limit $K\rightarrow \infty$ using collective measurements in general. If the QCRB is known to be saturated at the single-copy level then it will also be saturated on $K$ independent copies using only separable measurements by an additivity property of the quantum Fisher information matrix (see, e.g., \cite[Proposition 2.1]{Liu20}). 
 
\section{Overview of existing results}
\label{sec:existing}
In this work, we are interested in the saturability of the QCRB in multiparameter quantum estimation at the single-copy level. When $\rho_{\theta}$ is full rank then the QCRB is saturable on a single copy if and only if the full commutativity condition $[L_{\theta_j},L_{\theta_k}]=0$ holds for all $j,k$. If $\rho_{\theta}$ is not full rank then the full commutativity is no longer necessary. For pure quantum states  $\rho_{\theta}=|\psi\rangle \langle \psi|$, the QCRB can be saturated on a single copy if and only the average commutativity condition holds \cite{Matsumoto02,Pezze17},
$$
\mathrm{tr}(\rho_{\theta}[L_{\theta_j},L_{\theta_k}]) =0,\,\forall j,k.
$$
For general mixed density operators $\rho_{\theta}$, this condition is no longer sufficient in the single copy case but it remains necessary and sufficient for saturating the QCRB asymptotically in the multi-copy case (as $K \rightarrow \infty$) with collective measurements on the $K$ copies \cite{RJD16}. This makes achieving saturation experimentally challenging. For the single copy scenario, it was shown in \cite{Yang19} that partial commutativity of the SLD operators on the support of $\rho_{\theta}$ is necessary for saturability of the QCRB. Partial commutativity here is in the sense
\begin{equation}
\langle \psi_{m,\theta}|[L_{\theta_j},L_{\theta_k}]| \psi_{n,\theta}\rangle=0,\,\forall j,k=1,\ldots,p,\, m,n = 1,\ldots,r_{+}. \label{eq:partial-commutativity}
\end{equation}
 
To proceed further, the following definition will be required:
\begin{definition}
For a given POVM $\{E_k;k=1,2,\ldots,M\}$, a POVM operator (or element) $E_k$ is said to be {\rm regular} at $\theta$ if $\mathrm{tr}(\rho_{\theta} E_k) >0$, otherwise the POVM operator is said to be a {\em null} operator ($\mathrm{tr}(\rho_{\theta} E_k) =0$).
\end{definition}

The partially commutativity condition \eqref{eq:partial-commutativity} as a necessary condition follows from the following result \cite[Theorems 1 and 2]{Yang19}.
\begin{theorem}
\label{thm:POVM-characterization}
The QCRB is saturated at parameter value $\theta$ by a measurement corresponding to a POVM $\{E_k;k=1,2,\ldots,M\}$ if and only if:
\begin{enumerate}
\item If $E_k$ is a regular POVM operator then
\begin{equation}
E_k L_{\theta_l} |\psi_{n,\theta}\rangle = c^k_l E_k |\psi_{n,\theta}\rangle,\;\forall l=1,\ldots,p,\, n=1,\ldots,r_{+}, \label{eq:regular-cond}
\end{equation}
where $c^k_l$ is a real constant that depends on $k$ and $l$ but not on $n$.

\item If $E_k$ is a null  POVM operator then
\begin{equation}
E_k L_{\theta_l} |\psi_{n,\theta}\rangle = c^k_{lm} E_k L_{\theta_m}|\psi_{n,\theta}\rangle,\;\forall l,m=1,\ldots,p,\, n=1,\ldots,r_{+},\label{eq:null-cond}
\end{equation}
where $c^k_{lm}$ is a real constant that depends on $k$, $l$ and $m$ but not on $n$.

\end{enumerate}

\end{theorem}

\begin{corollary}
\cite[Theorem 3]{Yang19} If the conditions of Theorem \ref{thm:POVM-characterization} are satisfied then the partial commutativity condition \eqref{eq:partial-commutativity} holds.  
\end{corollary}

\section{Main results and discussion}
\label{sec:main}
In this section, the main theorem of the paper will be stated and proven. When sufficient conditions of the theorem are satisfied, the QCRB is saturated by a projective measurement that can be explicitly characterized. An example is also developed to illustrate the application of the theorem  to a non-full rank quantum state with two parameters.  

To get to the core arguments and methodology with minimal technicalities, the focus is on the finite-dimensional setting. However, it is reasonable to expect that the results will continue to hold, perhaps with the addition of some technical caveats, to infinite-dimensional separable Hilbert spaces when all operators have discrete countable spectra and eigenvectors (e.g., \cite{Yang19} allows an infinite-dimensional Hilbert space with a countable basis). It is also reasonable to expect that the results can be extended to continuous-variable quantum systems such as Gaussian quantum systems. We begin with the following lemma.
\begin{lemma}
\label{lem:regular-null-decom}
$E_k $ is a regular POVM operator if and only if $E_{k,++} \geq 0$ and $E_{k,++}\neq 0$. On the other hand, $E_k$ is a null POVM operator if and only if 
$$
E_k =\left[ \begin{array}{cc} 0 & 0 \\ 0 & E_{k,00} \end{array} \right],
$$
where $E_{k,00} \geq 0$. 
\end{lemma}
\begin{proof}
By using the representation \eqref{eq:rho-rep} for $\rho_{\theta}$, we have that
$
\mathrm{tr}(\rho_{\theta} E_k) =\mathrm{tr}(E_{k,++} \rho_{\theta,++}).
$
Since $\rho_{\theta,++}>0$ and diagonal in the basis $\mathcal{B}_{+,\theta}$, this quantity is positive and $E_k$ is regular if and only if $E_{k,++} \geq 0$ and $E_{k,++} \neq  0$. No other conditions are imposed on $E_{k,+0}$, $E_{k,0+}=E_{k,+0}^{\dag}$ and $E_{k,00}$ other than that they are such that $E_k \geq 0$. 

On the other hand, for a null POVM operator $\mathrm{tr}(\rho_{\theta} E_k)=0$ only if $E_{k,++}=0$. However, this is not sufficient since $E_k$ must also be positive semidefinite. Let $z$ be an $n_s$-dimensional complex vector, $z=[\begin{array}{cc} x^{\top} & y^{\top}\end{array}]$ with $x \in \mathbb{C}^{r_+}$ and $y \in\mathbb{C}^{r_0}$. Since $E_{k,++}=0$, it follows that $z^{\dag} E_k z =2\Re\{x^{\dag}E_{k,+0} y\} + y^{\dag} E_{k,00} y$ and therefore $z^{\dag} E_k z \geq 0$ for all $z$ if and only if $E_{k,+0}=0$. 
This proves the necessary and sufficient conditions for $E_k$ to be null. 
\end{proof}

Note that the conditions in Theorem \ref{thm:POVM-characterization} are actually statements about subspaces since conditions \eqref{eq:regular-cond} and \eqref{eq:null-cond}
hold independently of the index $n$. Indeed, it is immediately verified that  \eqref{eq:regular-cond} and \eqref{eq:null-cond} continue to hold when $|\psi_{n,\theta}\rangle$ is replaced by any $|\psi\rangle = \sum_{n=1}^{r_{+}} \lambda_k |\psi_{n,\theta}\rangle \in \mathcal{H}_{+,\theta}$ for any complex constants $\lambda_1,\ldots,\lambda_{r+}$.  
From these observations, the following statement can be extracted. 
\begin{lemma}
\label{lem:optimal-POVMs} The conditions of Theorem \ref{thm:POVM-characterization} can be stated equivalently as follows:
\begin{enumerate}
\item For a regular $E_k$,  \eqref{eq:regular-cond}  is equivalent to
\begin{equation}
E_k L_{\theta_l} P_{+,\theta} = c^k_l E_k P_{+,\theta}\;\forall l=1,\ldots,p. \label{eq:regular-cond-gen}
\end{equation}

\item For a null  $E_k$,  \eqref{eq:null-cond}  is equivalent to
\begin{equation}
E_k L_{\theta_l} P_{+,\theta} = c^k_{lm} E_k L_{\theta_m} P_{+,\theta} \;\forall l,m=1,\ldots,p. \label{eq:null-cond-gen}
\end{equation}

\end{enumerate}
\end{lemma}

The main result of this paper is Theorem \ref{thm:main} below \textcolor{blue}{[Note added: A corrigendum \cite{Nurd24c} has been included in Appendix \ref{app:corrigendum} that fixes a gap in Condition 2 of Theorem \ref{thm:main}. The correction is presented as Theorem \ref{thm:main-corrected}]}.  

\begin{theorem}
\label{thm:main}
Consider the following four conditions:
\begin{enumerate}
\item[1)] $[L_{\theta_l,++},L_{\theta_m,++}] =0$ for all $l,m=1,\ldots,p$.
\item[2)] For each $\theta$ there exists a unitary $U_{\theta} \in \mathbb{C}^{r_+ \times r_+}$ such that $U_{\theta}^{\dag}(\partial_l U_{\theta} - U_{\theta} V_{\theta}^{\dag} \partial_l V_{\theta})\rho_{\theta,++} + \rho_{\theta,++} (\partial_l U_{\theta} - U_{\theta} V_{\theta}^{\dag} \partial_l V_{\theta})^{\dag} U_{\theta}=0$ for  $l=1,\ldots,p$, where $\partial_l = \partial/\partial \theta_l$, $V_{\theta} =[\begin{array}{ccc} |\psi_{1,\theta} \rangle & \ldots & |\psi_{r_+,\theta} \rangle \end{array}]$, and $\rho_{\theta,++}$ is represented in the basis $\mathcal{B}_{+,\theta}$. 
\item[3)] $L_{\theta_l,+0} L_{\theta_m,+0}^{\dag}-L_{\theta_m,+0} L_{\theta_l,+0}^{\dag}=0$ for all $l,m=1,\ldots,p$.
\item[4)] There exists an $r_0 \times r_0$ complex  unitary matrix $W$ such that  all corresponding columns of  $L_{\theta_l,+0}W$ and $L_{\theta_m,+0}W$ with $l,m=1,\ldots,p$ are real scalar multiples of one another or the corresponding columns are simultaneously vanishing. That is, the $s$-th column of $L_{\theta_l,+0}W$ is either $\lambda_{lms} \in \mathbb{R}$ times the $s$-th column of $L_{\theta_m,+0}W$, or both columns are zero (i.e., all their elements are zero),  for all columns $s$ and $\forall l,m$.
\end{enumerate}
Then the following statements hold for saturability of the multiparameter QCRB in the single-copy setting: 
\begin{enumerate}
\item[i)] Conditions 1 and 2 are necessary and sufficient.
\item[ii)] Condition 3 is necessary and together with Condition 1 imply partial commutativity. 
\item[iii)] Conditions 1 and 4 are sufficient.
\end{enumerate}

When the sufficient conditions (Conditions 1 and 4) are satisfied, there exists an optimal projective measurement given by the POVM:
\begin{align*}
&\underbrace{\biggl\{ \left[\begin{array}{cc} \Pi_{\theta,1} & 0 \\ 0 & 0 \end{array} \right],\ldots, \left[\begin{array}{cc} \Pi_{\theta,\chi_{\theta}} & 0 \\ 0 & 0 \end{array} \right]\biggr\}}_{\hbox{Regular POVM operators}}  \bigcup \underbrace{\biggl\{ \left[\begin{array}{cc} 0 & 0 \\ 0 & W D_{00,1} W^{\dag} \end{array} \right],\ldots,\left[\begin{array}{cc} 0 & 0 \\ 0 & W D_{00,r_{0}} W^{\dag} \end{array} \right]\biggr\}}_{\hbox{Null POVM operators}},
\end{align*}
where $\Pi_{\theta_j}$ for $j=1,\ldots,\chi_{\theta}$ ($\chi_{\theta} \leq r_{+}$) are the (common) projection operators in the spectral decomposition of $L_{\theta_l,++}$ for $l=1,\ldots,p$, and $D_{00,j}$ is the $r_0\times r_0$ projection operator which is zero everywhere except for a 1 in row $j$ and column $j$ for $j=1,\ldots,r_0$.
\end{theorem}
\begin{remark}
\label{rem:cond-redundancy} Note that it can be verified by direct calculation that Condition 4 of the theorem on the existence of $W$ and on the columns of $L_{\theta_l,+0}W$ and $L_{\theta_m,+0}W$ implies the necessary Condition 3 of the theorem. Therefore, when Conditions 1 and 4 are satisfied, Condition 3 is redundant.
\end{remark}

\begin{proof}[Proof of Theorem \ref{thm:main}]
We begin by noting that the partial commutativity condition \eqref{eq:partial-commutativity} can be stated as
$$
P_{+,\theta} [L_{\theta_l},L_{\theta_m}]P_{+,\theta} =0\quad \forall l,m =1,\ldots,p.
$$
Using the four block decomposition \eqref{eq:decom} of $L_{\theta_j}$ as an observable, the identity above is equivalent to:
\begin{equation}
[L_{\theta_l,++},L_{\theta_m,++}]+(L_{\theta_l,+0}L_{\theta_m,+0}^{\dag} - L_{\theta_m,+0}L_{\theta_l,+0}^{\dag})=0. \label{eq:commutativity-block}
\end{equation}
The remaining steps of the proof are as follows. It will first be shown that Conditions 1 and 2 of the theorem are necessary and sufficient. By \eqref{eq:commutativity-block} this then implies that Condition 3  is also necessary. It will then be shown that Conditions 1 and 4 (Condition 3 is then implied, Remark \ref{rem:cond-redundancy}) are sufficient for the existence of a projective POVM such that \eqref{eq:regular-cond-gen} and \eqref{eq:null-cond-gen} hold, by explicitly constructing the POVM.

First recall from \S \ref{sec:prelim} that $\rho_{\theta,++}$ is diagonal in the basis $\mathcal{B}_{+,\theta}$ and from \eqref{eq:rho-decomposition} is given in this basis as $\rho_{\theta,++}=\mathrm{diag}(q_{1,\theta},\ldots,q_{r_+,\theta})$, where $ \mathrm{diag}(a_1,\ldots,a_{n})$ denotes a diagonal matrix with $a_{j}$ in row and column $j$. Introduce the shorthand notation $\partial_l x =\partial x/\partial \theta_l$ and $| \partial_l \psi_{k,\theta}\rangle = \partial_l |\psi_{k,\theta}\rangle$. Differentiating,  
$$
\partial_l \rho_{\theta,++}=\mathrm{diag}\left(\partial_l q_{1,\theta},\ldots,\partial_l q_{r_+,\theta}\right).
$$
Let $V_{\theta} =[\begin{array}{ccc} |\psi_{1,\theta} \rangle & \ldots & |\psi_{r_+,\theta} \rangle \end{array}]$. The right hand side (RHS) of the first equation in \eqref{eq:SLD-cond} can be expressed as, 
\begin{align}
\left. P_{+,\theta} \partial_l \rho_{\theta} P_{+,\theta} \right|_{\mathcal{H}_{+,\theta}} &=  \left. V_{\theta} V_{\theta}^{\dag} \partial_l (V_{\theta} \rho_{\theta,++} V_{\theta}^{\dag}) V_{\theta} V_{\theta}^{\dag} \right|_{\mathcal{H}_{+,\theta}}  \notag \\
&= \left. V_{\theta}(V_{\theta}^{\dag} \partial_l V_{\theta} \rho_{\theta,++}  + \partial_l \rho_{\theta,++} +\rho_{\theta,++} \partial_l  V_{\theta}^{\dagger} V_{\theta})V_{\theta}^{\dag} \right|_{\mathcal{H}_{+,\theta}} \notag\\
&= V_{\theta}^{\dag} \partial_l V_{\theta} \rho_{\theta,++}  + \partial_l \rho_{\theta,++} +\rho_{\theta,++} \partial_l V_{\theta}^{\dag} V_{\theta},
\label{eq:project-partial-rho}
\end{align}
where the last line gives the matrix representation of $\left. P_{+,\theta} \partial_l \rho_{\theta} P_{+,\theta} \right|_{\mathcal{H}_{+,\theta}}$ with respect to the basis $\mathcal{B}_{+,\theta}$.

For regular POVM operators $\{E_k\}_{k=1,\ldots,K}$ let $p_{k,\theta} = \mathrm{tr}(E_k \rho_{\theta})=\mathrm{tr}(E_{k,++} \rho_{\theta,++})$. Then $\{E_{k,++}\}_{k=1,\ldots,K}$ may be viewed as a POVM on $\mathcal{H}_{+,\theta}$ since they satisfy $\sum_{k=1}^K E_{k,++} = I_{r_+}$ and the probability $p_{k,\theta}$ only depends on $E_{k,++}$. Since $\rho_{\theta,++}>0$, a well-known result is that for a density operator $\rho'_{\theta}>0$ there is a  POVM that saturates the QCRB if and only if the SLDs $\{L'_{\theta_j}\}_{j=1,\ldots,p}$  satisfying $(L'_{\theta_j} \rho'_{\theta} + \rho'_{\theta}  L'_{\theta_j})/2 =\partial_j \rho'_{\theta}$  for $j=1,\ldots,p$ are mutually commuting, $[L'_{\theta_l},L'_{\theta_m}]=0$ for $l,m=1,\ldots,p$ \cite[\hbox{ref. [5]  and Appendix B.1}]{SYH20}. The RHS of the first equation of \eqref{eq:SLD-cond} does not give  $\partial_j \rho_{\theta,++}$ but rather the RHS of \eqref{eq:project-partial-rho}. However, the representation basis in $\mathcal{H}_{+,\theta}$ can be freely changed from $\mathcal{B}_{+,\theta}$ to a new basis $\mathcal{C}_{+,\theta}=\{|\widetilde{\psi}_{1,\theta}\rangle,\ldots,|\widetilde{\psi}_{r_+,\theta}\rangle \}$ defined by $[\begin{array}{ccc} |\widetilde{\psi}_{1,\theta}\rangle & \ldots & |\widetilde{\psi}_{r_+,\theta}\rangle \end{array}] = V_{\theta} U^{\dag}_{\theta}$, for some unitary $U_{\theta}\in \mathbb{C}^{r_+ \times r_+}$. Let $\varrho_{\theta,++}=U_{\theta} \rho_{\theta,++}U_{\theta}^{\dag}$ and  $\widetilde{L}_{\theta_l,++} = U_{\theta} L_{\theta_l,++} U_{\theta}^{\dag}$, using \eqref{eq:project-partial-rho} we obtain from \eqref{eq:SLD-cond},
\begin{align*}
(\widetilde{L}_{\theta_l,++} \varrho_{\theta,++} + \varrho_{\theta,++}  \widetilde{L}_{\theta_l,++})/2 
&  = U_{\theta} V_{\theta}^{\dag} \partial_l V_{\theta} \rho_{\theta,++} U_{\theta}^{\dag} + U_{\theta}\partial_l \rho_{\theta,++}U_{\theta}^{\dag} +U_{\theta} \rho_{\theta,++} \partial_l V_{\theta}^{\dag} V_{\theta}U_{\theta}^{\dag}.
\end{align*}
Finally, to write the RHS of the previous equality as 
$$
 \partial_l U_{\theta}\rho_{\theta,++} U_{\theta}^{\dag} + U_{\theta}\partial_l \rho_{\theta,++}U_{\theta}^{\dag} + U_{\theta} \rho_{,\theta,++}\partial_l U_{\theta}^{\dag}=\partial_l \varrho_{\theta,++},$$
we must have that
\begin{align}
U_{\theta}^{\dag}(\partial_l U_{\theta} - U_{\theta} V_{\theta}^{\dag} \partial_l V_{\theta})\rho_{\theta,++}   + \rho_{\theta,++} (\partial_l U_{\theta} - U_{\theta} V_{\theta}^{\dag} \partial_l V_{\theta})^{\dag} U_{\theta} &= 0. \label{eq:U-PDE}
\end{align}
Therefore, the commutativity $[L_{\theta_l,++},L_{\theta_m,++}]=U_{\theta}^{\dag} [\widetilde{L}_{\theta_l,++},\widetilde{L}_{\theta_m,++}] U_{\theta}=0$ for $l,m=1,\ldots,p$ and the existence of a unitary $U_{\theta}$ satisfying \eqref{eq:U-PDE} for $l=1,\ldots,p$ are {\em necessary and sufficient}. From \eqref{eq:commutativity-block} it follows that $L_{\theta_l,+0}L_{\theta_m,+0}^{\dag} - L_{\theta_m,+0}L_{\theta_l,+0}^{\dag}=0$ for all $l,m=1,\ldots,p$, thus Condition 3 of the theorem is also necessary.

Observe that by Lemma \ref{lem:regular-null-decom}, for a null POVM operator $E_k$ the condition \eqref{eq:null-cond-gen} reduces to
\begin{equation}
E_{k,00}(L_{\theta_l,+0}^{\dag} -c_{lm}^k L_{\theta_m,+0}^{\dag})=0\quad \forall l,m=1,\ldots,p. \label{eq:null-cond-red}
\end{equation}
Let $W \in \mathbb{C}^{r_0\times r_0}$ be a unitary matrix such that the conditions on the corresponding columns of $L_{\theta_l,+0}W$ and $L_{\theta_m,+0}W$ for $l,m=1,\ldots,p$ as stated in Condition 4 of the theorem are satisfied. Consider $E_{k,00}=WD_{k,00}W^{\dag}$ for some diagonal $D_{k,00}$ with $D_{k,00} \geq 0$. Then \eqref{eq:null-cond-red} is equivalent to  
\begin{equation}
D_{k,00} ((L_{\theta_l,+0}W)^{\dag} - c^k_{lm}(L_{\theta_l,+0} W)^{\dag}) =0, \label{eq:null-cond-aux}
\end{equation}
for all $l,m=1,\ldots,p$. By setting $D_{k,00}$ as specified in the theorem statement and  $c_{lm}^k=\lambda_{lmk}$, with $\lambda_{lmk}$ as defined in Condition 4, $\{E_{k,00}\}_{k=1,\ldots,r_0}$ forms a set of projective and mutually commuting operators with $\sum_{k=1}^{r_0}  E_{k,00} = I_{r_0}$. Hence the conditions are sufficient for \eqref{eq:null-cond-red}.

Now, we consider Condition 1. It will be shown that under Condition 1, a regular POVM operator corresponding to a projective measurement can be constructed that satisfies  \eqref{eq:regular-cond-gen}. By Lemma \ref{lem:regular-null-decom}, consider a regular POVM operator $E_{k}$,
$$
E_{k} = \left[ \begin{array}{cc} E_{k,++} & 0 \\ 0 & 0 \end{array} \right],
$$
with $E_{k,++} \geq 0$ and $E_{k,++} \neq 0$. Using this block decomposition, \eqref{eq:regular-cond-gen} reduces to
\begin{equation}
E_{k,++} L_{\theta_l,++} = c^k_j E_{k,++} \quad \forall l=1,\ldots,p. \label{eq:reg-++-cond}
\end{equation}
Since $[L_{\theta_l,++},L_{\theta_m,++}]=0$ by Condition 1, there is a common set of projection operators such that $L_{\theta_l,++}$ has the spectral decomposition
$
L_{\theta_l,++} =\sum_{k=1}^{\chi_{\theta}} \lambda_{lk} \Pi_{\theta,k},\;l=1,\ldots,p
$.
where $\chi_{\theta} \leq r_+$ and $\Pi_{\theta,k}$ are mutually commuting projection operators. $[\Pi_{\theta,k},\Pi_{\theta,j}]=0$ and $\Pi_{\theta,j}^2 = \Pi_{\theta,j}$ for all $j,k=1,\ldots,\chi_{\theta}$ such that $\sum_{k=1}^{\chi_{\theta}} \Pi_{\theta,k}= I_{r_+}$. 
By setting $E_{k,++}=\Pi_{\theta,k}$ for $k=1,\ldots,\chi_{\theta}$, we have that \eqref{eq:reg-++-cond} is satisfied with  $c^k_j = \lambda_{jk}$, $E_{k,++} \geq 0$ and $E_{k,++} \neq 0$ as required. By construction, the operator $E_k$ satisfies \eqref{eq:regular-cond-gen} and is a projection operator. Moreover, by construction, $\sum_{k=1}^{\chi_{\theta}} E_{k,++} = I_{r_+}$. Finally, let $E_k$ for $k=\chi_{\theta}+1,\ldots,\chi_{\theta}+r_0$ be the null POVM operators constructed in the previous paragraph. It follows that $\sum_{k=1}^{M} E_{k} = I_{n_s}$ for $M=\chi_{\theta}+r_0$ as required for a POVM. This completes the proof. 
\end{proof}

The conditions of the theorem are quite stringent. Condition 1 requires that $L_{\theta_l,++}$ and $L_{\theta_m,++}$ commute on the support subspace $\mathcal{H}_{+,\theta}$ for all $l,m=1,\ldots,p$. Condition 2 involves a system of coupled nonlinear partial differential equations that are required to have a unitary solution $U_{\theta}$.  Condition 3 is also necessary and is imposed by partial commutativity. When $r_+=1$, Conditions 1 and 3 in Theorem \ref{thm:main} are necessary and sufficient. In this case Condition 1 holds trivially and, by \eqref{eq:commutativity-block}, the necessary and sufficient average commutativity condition for pure states \cite{Matsumoto02} holds if and only if Condition 3 also holds. The next example illustrates a multiparameter quantum state that satisfies the conditions of the theorem.

The following example illustrates a multiparameter quantum state that satisfies the conditions of the theorem for all $\theta$ in its specified parameter set $\Theta$.  
\begin{example}\label{ex:saturable}
Consider the quantum state $\rho_{\theta}$ on $\mathcal{H}=\mathbb{C}^3$ (a three-level system or qutrit) parameterized by the vector $\theta=(\theta_1,\theta_2)$ in the parameter set $\Theta = (0,1)\times (0,1)$ given by:
$$
\rho_{\theta} = \left[\begin{array}{ccc} |d|^2(1-\theta_1) & 0 & (1-\theta_1)d\sqrt{1-|d|^2}e^{i\phi(\theta)} \\ 0 & \theta_1 & 0 \\(1-\theta_1)\bar{d}\sqrt{1-|d|^2}e^{-i\phi(\theta)} & 0 & (1-\theta_1)(1-|d|^2) \end{array}\right],  
$$
where $d$ is a complex number satisfying $0 <|d|<1$ and $\phi(\theta) = c_1 \theta_1 + c_2 \theta_2$ for some real non-zero constants $c_1$ and $c_2$.
The state is a mixture of two pure states and satisfies $\mathrm{rank}(\rho_{\theta})=2$ for all $\theta \in \Theta$. The projection operator to the null space can be computed explicitly to be
$$
P_{0,\theta} = \left[\begin{array}{ccc} 1-|d|^2  & 0 & -d\sqrt{1-|d|^2} e^{i\phi(\theta)} \\ 0 & 0 &  0 \\ -\bar{d} \sqrt{1-|d|^2}e^{-i\phi(\theta)} & 0 & |d|^2\end{array}\right],
$$
and so
$$
P_{+,\theta} = I-P_{0,\theta} =\left[\begin{array}{ccc} |d|^2  & 0 & d\sqrt{1-|d|^2} e^{i\phi(\theta)} \\ 0 & 1 &  0 \\ \bar{d} \sqrt{1-|d|^2}e^{-i\phi(\theta)} & 0 & 1-|d|^2\end{array}\right].
$$
We also have that
\begin{align*}
\lefteqn{\frac{\partial \rho_{\theta}}{\partial \theta_1}}\\
&=\left[\begin{array}{ccc} -|d|^2 & 0 & d\sqrt{1-|d|^2}(-1+ic_1(1-\theta_1))e^{i\phi(\theta)}  \\ 0 & 1 &  0 \\ \bar{d}\sqrt{1-|d|^2}(-1-ic_1(1-\theta_1))e^{-i\phi(\theta)} & 0 & -(1-|d|^2) \end{array}\right]
\end{align*}
and 
$$
\frac{\partial \rho_{\theta}}{\partial \theta_2}=\left[\begin{array}{ccc} 0 & 0 & ic_2(1-\theta_1)d\sqrt{1-|d|^2}e^{i\phi(\theta)}\\ 0 & 0 &  0 \\ -ic_2(1-\theta_1)d\sqrt{1-|d|^2}e^{-i\phi(\theta)} & 0 & 0 \end{array}\right].
$$

With some lengthy and tedious calculations it may be verified that
$$
P_{+,\theta} \frac{\partial \rho_{\theta}}{\partial \theta_1} P_{0,\theta} = \left(\frac{c_2}{c_1} \right)P_{+,\theta} \frac{\partial \rho_{\theta}}{\partial \theta_2} P_{0,\theta} \neq 0
$$
and 
$$
P_{+,\theta} \frac{\partial \rho_{\theta}}{\partial \theta_2} P_{+,\theta} = 0_{3 \times 3}.
$$
Using \eqref{eq:SLD-cond} it follows from the above identities (since $\rho_{\theta,++}>0$) that $L_{\theta_1,+0} = (c_2/c_1)L_{\theta_2,+0} \neq 0$ and $L_{\theta_2,++} = 0$, respectively. This is enough to verify Conditions 1 and 4 (the latter with $W=1$) of Theorem \ref{thm:main} (hence also Condition 3 by Remark \ref{rem:cond-redundancy}) for all $\theta \in \Theta$. Therefore, the QCRB can be saturated for this quantum state by the projective POVM specified in the theorem. This concludes the example. 
\end{example}

\section{Conclusion}
\label{sec:conclu} This paper has established, for finite-dimensional density operators, necessary and sufficient conditions for saturability of the QCRB in multiparameter single-copy 
estimation. In addition, new necessary conditions that imply partial commutativity and that become sufficient with the addition of another condition have also been obtained. The paper explicitly characterizes the measurement that saturates the QCRB when the sufficient conditions are satisfied, which turns out can always be chosen to be projective. As such the results make a significant advance towards understanding conditions for saturating the QCRB in the single-copy setting and, in particular, resolves the open problem of necessary and sufficient conditions \cite{HRZ22}.

The results have focused on quantum systems with a finite-dimensional Hilbert space in order to extract the essential ideas needed to address the problem. However, it is reasonable to anticipate that the methodology employed here can be suitably adapted to infinite-dimensional quantum systems such as continuous-variable quantum systems, in particular quantum Gaussian systems. They will be the subject of future research continuing from this one. 

\appendix

\section{Corrigendum to ``Saturability of the Quantum Cram\'{e}r-Rao Bound in Multiparameter Quantum Estimation at the Single-Copy Level"}
\label{app:corrigendum}
\begin{abstract}
A subtlety is found in the sufficiency of Condition 2 in Theorem 2 of the main paper since the associated system of PDEs can have non-unique unitary solutions and there can exist a unitary solution that does not correspond to saturation of the QCRB by the null POVM operators. 
This corrigendum fixes this gap and restores the necessary and sufficient conditions by replacing the original system of PDEs with a related stronger system of PDEs under a certain constraint. Condition 2 in Theorem 2 of the paper is corrected. Examples are given to illustrate instances where the corrected conditions are fulfilled. 
\end{abstract}

The corrigendum \cite{Nurd24c} uses the notation in the main paper. Theorem \ref{thm:main} states the following necessary and sufficient conditions:
\begin{enumerate}
\item[1)] $[L_{\theta_l,++},L_{\theta_l,++}]=0$ for $l,m=1,\ldots,p$. 

\item[2)] For each $\theta$ there exists a unitary $U_{\theta} \in \mathbb{C}^{r_+ \times r_+}$ such that $U_{\theta}^{\dag}(\partial_l U_{\theta} - U_{\theta} V_{\theta}^{\dag} \partial_l V_{\theta})\rho_{\theta,++} + \rho_{\theta,++} (\partial_l U_{\theta} - U_{\theta} V_{\theta}^{\dag} \partial_l V_{\theta})^{\dag} U_{\theta}=0$ for  $l=1,\ldots,p$, where $\partial_l = \partial/\partial \theta_l$, $V_{\theta} =[\begin{array}{ccc} |\psi_{1,\theta} \rangle & \ldots & |\psi_{r_+,\theta} \rangle \end{array}]$, and $\rho_{\theta,++}$ is represented in the basis $\mathcal{B}_{+,\theta}$.
\end{enumerate}
However, the author recently found a subtlety in the  sufficiency of Condition 2. This is because there can exist a unitary solution to this condition that does not necessarily correspond to saturation of the QCRB. For example, when $r_+=1$ or when $\rho_{\theta,++}=(1/r_+)I_{r_+}$ then $U_{\theta}=I_{r_+}$ satisfies Condition 2 (since $V_{\theta}$ is an isometry, $V_{\theta}^{\dag}V_{\theta}=I_{r_+}$, and therefore $\partial_l V_{\theta}^{\dag}V_{\theta}$ is skew-hermitian) without imposing any constraints on $V_{\theta}$ as would be required. The conditions on $V_{\theta}$ relate to null POVM operators that saturate the QCRB (discussed further below). In the following the correct necessary and sufficient conditions will be recovered by modifying Condition 2.

Let $Y_{\theta} =\left[\begin{array}{ccc} |\phi_{1,\theta}\rangle & \ldots & |\phi_{r_0,\theta}\rangle \end{array} \right] \in \mathbb{C}^{n_s \times r_0}$ have columns that are the orthonormal basis of $\mathcal{H}_{0,\theta}$. Note that $V_{\theta}^{\dag}Y_{\theta}=0$ and $Y_{\theta}$ is also an isometry, $Y_{\theta}^{\dag}Y_{\theta}=I_{r_0}$. By following a similar calculation to the proof of Theorem \ref{thm:main}, it can be shown that  $\left. P_{+,\theta}\partial_l (V_{\theta} \rho_{\theta,++}  V_{\theta}^{\dag}) P_{0,\theta}\right|_{\mathcal{H}_{0,\theta}} = V_{\theta}(\rho_{\theta,++} \partial_l V_{\theta}^{\dag}Y_{\theta} )Y_{\theta}^{\dag}$. Therefore, in the basis $\mathcal{B}_{k,\theta}$ for $\mathcal{H}_{k,\theta}$ and $k \in \{+,0\}$, $L_{\theta_l,+0} = 2\partial_l V_{\theta}^{\dag}Y_{\theta}$ for $l=1,\ldots,p$  \cite{Nurd24b}. Optimal null POVM operators that saturate the QCRB depend only on $L_{\theta_l,+0}$ for $l=1,\ldots,p$ (by Lemma \ref{lem:optimal-POVMs} and Eq. \eqref{eq:null-cond-red}) and thus in turn depend only on $V_{\theta}$ since $\partial_l V_{\theta}$  and $Y_{\theta}$ are determined by $V_{\theta}$. 

Since the same set of null operators would saturate $\rho_{\theta}=V_{\theta}\rho_{\theta,++}V_{\theta}^{\dag}$ for a fixed $V_{\theta}$ and any choice of $\rho_{\theta,++}$, the equality involving the partial derivatives in Condition 2 must hold for all diagonal density operators $\rho_{\theta,++}$. By the observation that $X \in \mathbb{C}^{r_+ \times r_+}$ satisfies $X \rho_{\theta,++} +  \rho_{\theta,++}X^{\dag} =0$ for all $\rho_{\theta,++}$ if and only if $X=iD$ for some diagonal matrix $D \in \mathbb{R}^{r_+ \times r_+}$, this implies that there exist some, generally $\theta$-dependent, diagonal matrices $D_{l,\theta} \in \mathbb{R}^{r_+ \times r_+}$ for $l=1,\dots,p$ such that the system of PDEs: 
\begin{equation}
\partial_l U_{\theta} = U_{\theta} (V_{\theta}^{\dag} \partial_l V_{\theta}+iD_{l,\theta}),\;   l=1,\ldots,p, \label{eq:U-PDEs}
\end{equation}
has a unitary solution $U_{\theta}$. In general, the choice of $D_{1,\theta},\ldots,D_{p,\theta}$ may not be unique; see Example \ref{ex:L-CSS}. 
\newpage 
Now, consider the following modified condition: 
\begin{enumerate}
\item[2')] There exist diagonal matrices $D_{1,\theta},\ldots,D_{p,\theta} \in \mathbb{R}^{r_+ \times r_+}$ and a unitary $U_{\theta} \in \mathbb{C}^{r_+ \times r_+}$ solving the system of PDEs \eqref{eq:U-PDEs} such that $\widetilde{V}_{\theta} = V_{\theta}U_{\theta}^{\dag}$ satisfies \begin{eqnarray}
E_{k,00} Y_{\theta}^{\dag} (\partial_l \widetilde{V}_{\theta}-c^k_{lm} \partial_m \widetilde{V}_{\theta})=0,\, l,m=1,\ldots,p, \label{eq:null-saturation}
\end{eqnarray}
for some POVM $\{E_{k,00}\}_{k=1,\ldots,\nu_{\theta}}$ on $\mathcal{H}_{0,\theta}$ ($\sum_{k=1}^{\nu_{\theta}}E_{k,00}=I_{r_0}$) and real constants $c^k_{lm}$. 
\end{enumerate}
When \eqref{eq:U-PDEs} has a unitary solution $U_{\theta}$ then the original Condition 2 is fulfilled and, from the proof of Theorem \ref{thm:main}, $\varrho_{\theta,++}=U_{\theta} \rho_{\theta,++}U_{\theta}^{\dag}$ satisfies $(\widetilde{L}_{\theta_l,++}\varrho_{\theta,++} + \varrho_{\theta,++} \widetilde{L}_{\theta_l,++})/2 =\partial_l  \varrho_{\theta,++}$ for all $l$ and its associated QCRB is saturated if and only if Condition 1 holds. In this case the regular POVM operators can be chosen to be $V_{\theta}\Pi_{\theta,l}V_{\theta}^{\dag}$ for $l=1,\ldots,\chi_{\theta}$ with $\Pi_{\theta, l}$ and $\chi_{\theta}$ as given in Theorem \ref{thm:main}; as shown in main paper these regular POVM operators satisfy the relevant condition in Lemma \ref{lem:optimal-POVMs}. Since $\widetilde{L}_{\theta_l,+0} =U_{\theta} L_{\theta_l,+0} = \partial_l \widetilde{V}_{\theta}^{\dag} Y_{\theta}$ for $l=1,\ldots,p$ and $U_{\theta}$ is invertible, by Lemma \ref{lem:optimal-POVMs} and Eq. \eqref{eq:null-cond-red} for the null POVM operators to saturate the QCRB it is necessary and sufficient that \eqref{eq:null-saturation} holds for some real constants $c^k_{lm}$. Moreover, since the optimal (saturating) regular POVM operators satisfy $\sum_{k=1}^{\nu_{\theta}} V_{\theta} \Pi_{\theta,k} V_{\theta}^{\dag} = \left[\begin{array}{cc} I_{r_+} & 0 \\ 0 & 0 \end{array} \right]$ it follows necessarily that $\sum_{k=1}^{\nu_{\theta}} E_{k,00} = I_{r_0}$. 

For the case where $D_{l,\theta}=0$ for all $l$ the following holds.
\begin{lemma}\label{lem:D-vanishing} If $D_{l,\theta}=0$ for all $l$ then under Condition 2' $P_{+,\theta} \partial_l \widetilde{V}_{\theta} =0$ for all $l$. 
\end{lemma}
\begin{proof}
Since $D_{l,\theta}=0$ for all $l$, under Condition 2' $V_{\theta}^{\dag} \partial_l V_{\theta} = U_{\theta}^{\dag} \partial_l U_{\theta}$, $l=1,\ldots,p$, and
\begin{align*}
\partial_l V_{\theta} &= V_{\theta}V_{\theta}^{\dag} \partial_l V_{\theta} +  Y_{\theta}Y_{\theta}^{\dag} \partial_l V_{\theta} \\  
&=   \tilde{V}_{\theta}^{\dag} \partial_l U_{\theta} +  Y_{\theta}Y_{\theta}^{\dag} (\partial_l \tilde{V}_{\theta}-V_{\theta} \partial_l U_{\theta}^{\dag})U_{\theta}. 
\end{align*}
Since $V_{\theta} = \widetilde{V}_{\theta}U_{\theta}$ and $U_{\theta}$ is invertible, it follows from the above that $\partial_l \widetilde{V}_{\theta} = Y_{\theta}Y_{\theta}^{\dag}(\partial_l \widetilde{V}_{\theta}-V_{\theta} \partial_l U_{\theta}^{\dag})$. Since $Y_{\theta}^{\dag} V_{\theta}=0$,  $\partial_l \widetilde{V}_{\theta} = Y_{\theta}Y_{\theta}^{\dag} \partial_l \widetilde{V}_{\theta} \Leftrightarrow (I- Y_{\theta}Y_{\theta}^{\dag}) \partial_l \widetilde{V}_{\theta} =0$. That is, $P_{+,\theta} \partial_l \widetilde{V}_{\theta} =0$ for all $l$. 
\end{proof}

Finally, since the matrix $W$ in Condition 4 of Theorem \ref{thm:main}) can more generally depend on $\theta$ (without any change to the proof), the theorem can be corrected as the following. 
\begin{theorem} \label{thm:main-corrected}(Correction to Theorem \ref{thm:main}) 
Consider Condition 1, Condition 2' and the following two conditions:
\begin{enumerate}
\item[3)] $L_{\theta_l,+0} L_{\theta_m,+0}^{\dag}-L_{\theta_m,+0} L_{\theta_l,+0}^{\dag}=0$ for all $l,m=1,\ldots,p$.
\item[4)] There exists a unitary $W_{\theta} \in \mathbb{C}^{r_0\times r_0}$ such that  all corresponding columns of  $L_{\theta_l,+0}W_{\theta}$ and $L_{\theta_m,+0}W_{\theta}$ with $l,m=1,\ldots,p$ are real scalar constant multiples of one another or the corresponding columns are simultaneously vanishing. That is, the $s$-th column of $L_{\theta_l,+0}W_{\theta}$ is either $\lambda_{lms} \in \mathbb{R}$ times the $s$-th column of $L_{\theta_m,+0}W_{\theta}$, or both columns are zero (i.e., all their elements are zero),  for all columns $s$ and $\forall l,m$.
\end{enumerate}
Then the following statements hold for saturability of the multiparameter QCRB in the single-copy setting: i) Conditions 1 and 2' are necessary and sufficient, ii) Condition 3 is necessary and together with Condition 1 imply partial commutativity, and iii) Conditions 1 and 4 are sufficient.

When the sufficient conditions (Conditions 1 and 4) are satisfied, there exists an optimal projective measurement given by the POVM:
\begin{align*}
&\underbrace{\biggl\{ \left[\begin{array}{cc} \Pi_{\theta,1} & 0 \\ 0 & 0 \end{array} \right],\ldots, \left[\begin{array}{cc} \Pi_{\theta,\chi_{\theta}} & 0 \\ 0 & 0 \end{array} \right]\biggr\}}_{\hbox{Regular POVM operators}} \\
&\quad\bigcup \underbrace{\biggl\{ \left[\begin{array}{cc} 0 & 0 \\ 0 & W_{\theta} D_{00,1} W_{\theta}^{\dag} \end{array} \right],\ldots,\left[\begin{array}{cc} 0 & 0 \\ 0 & W_{\theta} D_{00,r_{0}} W_{\theta}^{\dag} \end{array} \right]\biggr\}}_{\hbox{Null POVM operators}},
\end{align*}
where $\Pi_{\theta_j}$ for $j=1,\ldots,\chi_{\theta}$ ($\chi_{\theta} \leq r_{+}$) are the (common) projection operators in the spectral decomposition of $L_{\theta_l,++}$ for $l=1,\ldots,p$, and $D_{00,j}$ is the $r_0\times r_0$ projection operator which is zero everywhere except for a 1 in row $j$ and column $j$ for $j=1,\ldots,r_0$.
\end{theorem}

We conclude by exhibiting some examples satisfying Conditions 1 and 2' so that the QCRB is saturated. Note that in all of these examples $D_{l,\theta}$ can be chosen to be  $D_{l,\theta}=0$ for all $l=1,\ldots,p$. 

\begin{example}(\cite[Example 1]{Nurd24b}) Suppose that $\mathcal{H}_{+,\theta}=\mathcal{H}_{+}$ does not depend on $\theta$ and $V_{\theta} =B_{+} S_{\theta}$ where the columns of $B_{+} \in \mathbb{C}^{n_s \times r_+}$ are a fixed ($\theta$-independent) orthonormal basis of $\mathcal{H}_{+}$ and $S_{\theta}$ is some unitary-valued function such that $\partial_l S_{\theta}^{\dag}S_{\theta}$ is diagonal, $l=1,\ldots,p$. Then Conditions 1 and 2' are satisfied with  $U_{\theta}=S_{\theta}$ and for an arbitrary POVM $\{E_{k,00}\}_{k=1,\ldots,\nu_{\theta}}$ since $L_{\theta_l,0+}$ vanishes for all $l$.
\end{example}

\begin{example} \label{ex:L-CSS} Consider again Example \ref{ex:saturable}. For this example,
$V_{\theta}=\left[\begin{array}{cc} 0 & d e^{i (c_1\theta_1+c_2\theta_2)} \\ 1 & 0 \\ 0 & \sqrt{1-|d|^2} \end{array}\right]$, \\ $Y_{\theta} = [\begin{array}{ccc} \sqrt{1-|d|^2} & 0 & \bar{d}e^{-i(c_1\theta_1+c_2\theta_2)} \end{array}]^{\top}$, $U_{\theta}=\left[\begin{array}{cc} 1 & 0 \\ 0 & e^{i|d|^2 (c_1\theta_1+c_2\theta_2)}\end{array} \right]$ is a unitary solution of \eqref{eq:U-PDEs} (computed in \cite[Example 2]{Nurd24b}), and $\widetilde{V}_{\theta} = V_{\theta}U_{\theta}^{\dag} =\left[\begin{array}{cc} 0 & d e^{i (1-|d|^2)(c_1\theta_1+c_2\theta_2)} \\ 1 & 0 \\ 0 & \sqrt{1-|d|^2} e^{-i|d|^2(c_1\theta_1+c_2\theta_2)} \end{array}\right]$. Therefore, $\partial_l \widetilde{V}_{\theta} = i c_l f(\theta) \left[\begin{array}{cc} 0 & \sqrt{1-|d|^2} e^{i (c_1\theta_1+c_2\theta_2)} \\ 0 & 0 \\ 0 & -\bar{d}\end{array}\right]$ for $l=1,2$, where $f(\theta)= d \sqrt{1-|d|^2}e^{-i|d|^2 (c_1\theta_1+c_2\theta_2)}$.  Since $D_{l,\theta}=0$ for $l=1,2$, it can be directly verified that $P_{+,\theta} \partial_l \widetilde{V}_{\theta} = V_{\theta}V_{\theta}^{\dag} \partial_l \widetilde{V}_{\theta}=0$ for $l=1,2$ (Lemma \ref{lem:D-vanishing}). It has been shown in Example \ref{ex:saturable} that the QCRB is saturated for this example  thus Conditions 1 and 2' are fulfilled.

There is an alternative choice of $D_{l,\theta}$ for this example that is not necessarily vanishing. Since $\partial_l V_{\theta}^{\dag} V_{\theta}$ is diagonal and skew-hermitian, one can set $D_{l,\theta}=i \partial_l V_{\theta}^{\dag} V_{\theta}$ for all $l$. For this choice, $U_{\theta}=I_{r_+}$, $\widetilde{V}_{\theta}=V_{\theta}$ and Condition 2' becomes the condition \eqref{eq:null-cond-red} for $L_{\theta_l,+0}$, $l=1,2$. 
\end{example}

\begin{example} Suppose that $V_{\theta}$ is such that $\partial_l V_{\theta}^{\dag}V_{\theta}=0$, $l=1,\ldots,p$. Then $U_{\theta}=I_{r_+}$ is a unitary solution to \eqref{eq:U-PDEs} and $\widetilde{V}_{\theta}=V_{\theta}$. Condition 1 is satisfied since $L_{\theta_l,++}$ is diagonal in the basis $\mathcal{B}_{+,\theta}$ for all $l=1,\ldots,p$. If $\widetilde{V}_{\theta}$ satisfies \eqref{eq:null-saturation} then Condition 2' is satisfied.
\end{example}

\bibliographystyle{ieeetran}
\bibliography{mpqe}

\end{document}